\documentclass[journal,onecolumn,twoside]{IEEEtran}
\usepackage{ifpdf}
\usepackage{cite}
\usepackage[pdftex]{graphicx}
\usepackage{algorithmic}
\usepackage{array}
\usepackage{stfloats}
\usepackage{color}
\usepackage{amsmath}
\usepackage{amsfonts}
\usepackage{amssymb}
\usepackage{amsthm}
\usepackage{tikz}
\usepackage{bm}
\usepackage{multirow}
\usepackage{makecell}
\usepackage{mathdots}
\usepackage{xtab}
\usepackage{caption}
\usepackage{booktabs}
\addtocounter{MaxMatrixCols}{10}
\theoremstyle{plain}

\newtheorem{theorem}{Theorem}
\newtheorem{lemma}[theorem]{Lemma}

\newtheorem{corollary}[theorem]{Corollary}
\theoremstyle{definition}

\newtheorem{construction}{Construction}
\newtheorem{remark}{Remark}

\begin{document}
\title{New Centralized MSR Codes With Small Sub-packetization}
\author{ Yaqian Zhang

\thanks{Yaqian Zhang is with School of Electronic Information and Electrical Engineering, Shanghai Jiao Tong University, Shanghai 200240, China, e-mail: zhangyq9@sjtu.edu.cn}

}
\maketitle

\begin{abstract}
Centralized repair refers to repairing $h\geq 2$ node failures using $d$ helper nodes in a centralized way, where the repair bandwidth is counted by the total amount of data downloaded from the helper nodes.
A centralized MSR code is an MDS array code with $(h,d)$-optimal repair for some $h$ and $d$.
In this paper, we present several classes of centralized MSR codes with small sub-packetization.
At first, we construct an alternative MSR code with $(1,d_i)$-optimal repair for multiple repair degrees $d_i$ simultaneously.
Based on the code structure, we are able to construct a centralized MSR code with $(h_i,d_i)$-optimal repair property for all possible $(h_i,d_i)$ with $h_i\mid (d_i-k)$ simultaneously. The sub-packetization is no more than ${\rm lcm}(1,2,\ldots,n-k)(n-k)^n$, which is much smaller than a previous work given by Ye and Barg ($({\rm lcm}(1,2,\ldots,n-k))^n$).
Moreover, for general parameters $2\leq h\leq n-k$ and $k\leq d\leq n-h$, we further give a centralized MSR code enabling $(h,d)$-optimal repair with sub-packetization smaller than all previous works.
\end{abstract}

 \IEEEpeerreviewmaketitle
\section{Introduction}
Erasure codes are extensively used in large-scale distributed storage systems to protect the data reliability. Typically, a data file is encoded into a codeword and stored across $n$ storage nodes with each node storing one coordinate. Each node may contain an array of symbols, the size of which is called {\it node capacity} or {\it sub-packetization}.
The system should satisfy that any $k$ nodes are able to retrieve the original file.
When some node fails, the failed node should be repaired by downloading data from some surviving nodes ( helper nodes), where the number of helper nodes is called {\it repair degree}. An important metric to measure the repair efficiency is called {\it repair bandwidth}, i.e. the total amount of data downloaded during the repair process.
In a remarkable work \cite{Dimakis2011}, Dimakis et al. derived a cut-set bound that gives a tradeoff between the storage overhead and repair bandwidth.
Codes with parameters lying on this tradeoff curve are called regenerating codes.
An extreme point of the tradeoff curve gives a subclass of regenerating codes, called minimum storage regenerating (MSR) codes, which achieve the optimal tradeoff between redundancy and reliability.

MSR codes belong to a subclass of MDS codes, namely, MDS array codes \cite{arraycode}.
An $(n,k,\ell)$ MDS array code is formed by a set of length-$n\ell$ vectors $(\bm c_1,...,\bm c_n)\in F^{n\ell}$, where $F$ is the defining field of the code and  $\ell$ is the sub-packetization level.
It satisfies that any $k$ coordinates can be seen as information coordinates and
each coordinate $\bm c_i\in F^\ell$ is stored in one storage node $i$ for $i\in[n]$.
Indeed, MDS array codes with repair bandwidth achieving the cut-set bound indicated by \cite{Dimakis2011} are exactly MSR codes.

MSR codes typically deal with single node failures and have been extensively studied in the literature \cite{Cadambe2013,Rashmi,zigzag,Wang2016,Ye-highrate,Liu-2023-multi-d,Ye-access,Li-2017,Liu,Rawat,Li2021,Guruswami-2020}, wherein \cite{Ye-access,Li-2017,Liu} additionally consider the I/O cost, \cite{Rawat,Li2021,Guruswami-2020} focus on tradeoff between sub-packetization and repair bandwidth, and codes with multiple repair degrees are given in \cite{Ye-highrate,Liu-2023-multi-d}.
Besides, inspired by the idea of ``vectorization'', repairing scalar MDS codes ($\ell=1$), such as Reed-Solomon codes, is studied in \cite{Guruswami-rs, scalarMDS} to reduce the repair bandwidth.

In some practical scenarios, repairing multiple node failures is the usual case, such as Total Recall \cite{TotalRecall}. There are mainly two models for multi-node repair. One is the centralized repair model \cite{Wang-centralized}, where it assumes a data center to complete all repairs and the repair bandwidth measures the total amount of data downloaded from the helper nodes. The other is the cooperative repair model \cite{Hu-cooperative}, where the failed nodes are repaired in a distributed and cooperative way and the repair bandwidth includes both downloaded data from helper nodes and communicated data between failed nodes.
Cut-set bounds for MDS array codes under the two repair models are respectively obtained in \cite{Wang-centralized} and \cite{Hu-cooperative,Shum}, and MDS array codes with repair bandwidth achieving the corresponding cut-set bounds are called {\it centralized MSR codes} and {\it cooperative MSR codes}. In this paper,
we focus on centralized MSR codes.

In the centralized repair, suppose $h$ ($2\leq h\leq n-k$) nodes indexed by $\mathcal{H}\subseteq[n]$ with $|\mathcal{H}|=h$ are failed, i.e., the data $\bm{c}_i, i\in\mathcal{H}$ are erased. Then there assumes a data center to repair all the failed nodes. More precisely, the data center connects to any $d$ ($k\leq d\leq n-h$) helper nodes indexed by $\mathcal{R}\subseteq[n]\setminus\mathcal{H}$ with $|\mathcal{R}|=d$, and downloads $\beta(\mathcal{H},\mathcal{R})$ symbols in $F$ from each of the $d$ helper nodes. Then the repair bandwidth $\gamma(\mathcal{H},\mathcal{R})$ is counted by the total amount of data downloaded from the helper nodes by the data center. That is, $\gamma(\mathcal{H},\mathcal{R})=d\cdot\beta(\mathcal{H},\mathcal{R})$.
The cut-set bound \cite{Wang-centralized} indicates that
\begin{equation}\label{cut-set-bound}
\beta(\mathcal{H},\mathcal{R})\geq\frac{h\ell}{d-k+h}, ~~~~ \gamma(\mathcal{H},\mathcal{R})\geq\frac{dh\ell}{d-k+h}.
\end{equation}
If an MDS array code $\mathcal{C}$ achieves the cut-set bound in (\ref{cut-set-bound}) with equality for any $h$-subset $\mathcal{H}\subseteq[n]$ and any $d$-subset $\mathcal{R}\subseteq[n]\setminus\mathcal{H}$, we say $\mathcal{C}$ has $(h,d)$-optimal repair property, and $\mathcal{C}$ is a centralized MSR code.
In particular, when $h=1$, the code $\mathcal{C}$ with $(1,d)$-optimal repair property is an MSR code.

\subsection{Related works}
In the literature, the work \cite{Wang-centralized} presented a centralized repair scheme for MDS array codes which is defined over a large enough finite filed and the repair is for special values of $d$. Authors of \cite{Wang2019} developed a transformation from codes with single node repair to that with multiple failures, and presented product-matrix codes with optimal repair of low code rate.
For any $2\leq h\leq n-k$, $k\leq d\leq n-h$ and high code rate, an excellent work by Ye and Barg \cite{Ye-highrate} gave an explicit construction of centralized MSR codes with sub-packetization $({\rm lcm}(d-k+1,d-k+2,\ldots,d-k+h))^n$. They also extend it to a centralized MSR code with sub-packetization $({\rm lcm}(1,2,\ldots,n-k))^n$  possessing the $(h_i,d_i)$-optimal repair property for all $2\leq h_i\leq n-k$ and $k\leq d_i\leq n-h_i$ simultaneously.
In a recent work \cite{Li-arxiv}, Li et al. gave a construction of centralized MSR code to reduce the sub-packetization to $(\frac{d-k+h}{h})^n$ when $h\mid(d-k)$.
On the other hand, the work \cite{Ye-2018} proved MDS array codes with optimal repair bandwidth in cooperative repair must achieve optimal repair bandwidth in centralized repair. That is, a cooperative MSR code can be easily converted to a centralized one.
We list two best results of cooperative MSR codes as follows. The work \cite{Ye2020} presented cooperative MSR codes with sub-packetization $(d-k+h)(d-k+1)^n$. When $d=k+1$, authors of \cite{Liu2023} proved that the sub-packetization can be reduced to $2^n$ when $(h+1)\mid2^n$ and $(2\rho+1)2^n$ when $h+1=(2\rho+1)2^t$ for some integers $\rho$ and $t$.

\subsection{Our contribution}

We present new constructions of centralized MSR codes with small sub-packetization.
At first, we give an alternative MSR code with $(1,d_i)$-optimal repair for multiple repair degrees $d_i$ simultaneously.
Despite that the sub-packetization of the code is larger than the best known result given in \cite{Liu-2023-multi-d}, however, it can be used to construct codes with multiple node erasures.
More specifically, based on the code structure, we can construct a centralized MSR code with $(h_i,d_i)$-optimal repair for all possible $(h_i,d_i)$ with $h_i\mid (d_i-k)$ simultaneously. The sub-packetization is no more than ${\rm lcm}(1,2,\ldots,n-k)(n-k)^n$, which is much smaller than that in the work \cite{Ye-highrate}.
Moreover, for general parameters $2\leq h\leq n-k$ and $k\leq d\leq n-h$, we further give a centralized MSR code enabling $(h,d)$-optimal repair with sub-packetization smaller than the works \cite{Ye-highrate, Ye2020, Li-arxiv}. In particular, the works \cite{Li-arxiv} and \cite{Ye2020} can be seen as two extreme cases of our code constructions.
Besides, for special values of $h$ and $d$ with $(d-k)\mid h$  and $(\frac{h}{d-k}+1)\mid 2^n$, we show that the sub-packetization can be further reduced.
Finally, the code is extended to all possible $h$ and $d$ simultaneously.
Comparisons of our codes and previous works are given in Table \ref{tab1}.

\begin{table*}
\begin{center}
\caption{High rate centralized MSR codes with $(h,d)$-optimal repair.}
\label{tab1}
\begin{tabular}{| c | c | c |}
\hline
Ref. & repair pattern $(h,d)$ ($2\leq h\leq n-k$, $k\leq d\leq n-h$) & sub-packetization $\ell$  \\ \hline
\cite{Ye-highrate} & a pair of $(h,d)$ & $({\rm lcm}(d-k+1,d-k+2,\ldots,d-k+h))^n$  \\
\hline
\cite{Ye2020} & a pair of  $(h,d)$ & $(d-k+h)(d-k+1)^n$  \\
\hline
\cite{Li-arxiv} & a pair of  $(h,d)$ with $h\mid (d-k)$ & $(\frac{d-k+h}{h})^n$  \\
\hline
\cite{Ye-highrate} & pairs of $(h,d)$ for all possible $h$ and $d$ & $({\rm lcm}(1,2,\ldots,n-k))^n$  \\
\hline
This paper, Theorem \ref{thm3}& a pair of  $(h,d)$ & $(\frac{d-k+h}{\delta})(\frac{d-k+\delta}{\delta})^n$ where $\delta={\rm gcd}(h,d-k)$  \\
\hline
This paper, Theorem \ref{thm-hadamard}& a pair of  $(h,d)$, $(d-k)\mid h$  and $(\frac{h}{d-k}+1)\mid 2^n$ & $2^n$  \\
\hline
This paper, Corollary \ref{cor} & pairs of $(h,d)$ for all possible $h$, $d$ with $h\mid(d-k)$ & ${\rm lcm}(1,2,\ldots,n-k)\cdot(n-k)^n$  \\
\hline
This paper, Corollary \ref{cor2} & pairs of $(h,d)$ for all possible $h$, $d$ & ${\rm lcm}(1,2,\ldots,n-k)\cdot(n-k)^n$  \\
\hline
\end{tabular}
\end{center}
\end{table*}

\subsection{Organization}
The rest of the paper is organized as follows. Section \ref{sec-code1} gives an alternative MSR code with multiple repair degrees. Section \ref{sec-code2} presents a centralized MSR code enabling $(h_i,d_i)$-optimal repair for multiple $h_i, d_i$ with $h_i\mid(d_i-k)$ simultaneously. Section \ref{sec-code3} gives a centralized MSR code enabling $(h,d)$-optimal repair for a single pair of $(h,d)$ with general parameters. Section \ref{sec-code4} further extends it to all possible parameters $h$ and $d$ simultaneously.

\section{An alternative construction of MSR codes with multiple repair degrees}\label{sec-code1}
Throughout this paper, we use $[n]$ to denote the set of integers $\{1,2,...,n\}$ for a positive integer $n$, and denote $[i,j]=\{i,i+1,...,j\}$ for two integers $i<j$.
Let $k<d_1<d_2<\cdots<d_m\leq n-1$, we give a new construction of MSR codes with $(1,d_i)$-optimal repair for multiple $d_i$, $i\in[m]$, simultaneously. Denote $s_i=d_i-k+1$ for $i\in[m]$ and $s={\rm lcm} (s_1,s_2,\ldots,s_{m-1})$. The new code $\mathcal{C}_1$ is an $(n,k,\ell=s\cdot s_m^n)$ MDS array code over a finite field $F$. For each codeword $\bm{c}\in\mathcal{C}_1$, write $\bm{c}=(\bm{c}_1,\bm{c}_2,\ldots,\bm{c}_n)$, where $\bm{c}_i=(c_{i,0},\ldots,c_{i,\ell-1})\in F^\ell$ for $i\in[n]$.

$\forall a\in[0,s_m^n-1]$,  write $a=(a_1,a_2,\ldots,a_n)$ where $a$ has the $s_m$-ary expansion $\sum_{i=1}^n a_i\cdot s_m^{i-1}$ with $a_{i}\in[0,s_m-1]$.
Given $a\in[0,s_m^n-1]$, for some $i\in[n]$ and $v\in[0,s_m-1]$, define $a(i,v)$ to be the integer with representation $(a_1,\ldots,a_{i-1},v,a_{i+1},\ldots,a_{n})$.
More generally, for a subset $\mathcal{X}\subseteq[n]$ and a vector $\bm{v}\in[0,s_m-1]^{|\mathcal{X}|}$, define $a(\mathcal{X},\bm{v})$ to be the integer $a'$ with $a'|_{\mathcal{X}}=\bm{v}$ and $a'|_{[n]\setminus\mathcal{X}}=a|_{[n]\setminus\mathcal{X}}$ where $a'|_{\mathcal{X}}$ denotes the vector representation of $a'$ restricted on the digits in $\mathcal{X}$.
For every integer $\tau\in[0,\ell-1]$,  $\tau$ has a unique representation as $\tau=b\cdot s_m^n+\sum_{i=1}^n a_i\cdot s_m^{i-1}$ with $b\in[0,s-1]$ and $a_{i}\in[0,s_m-1]$ for $i\in[n]$. Then we can write $\tau=(a_1,a_2,\ldots,a_n,b)=(a,b)$ where $a=(a_1,a_2,\ldots,a_n)\in[0,s_m^n-1]$. In the following, we use pairs $(a,b)$ to represent integers in $[0,\ell-1]$ and use the notation $\oplus$ to denote addition modulo $s_m$, that is, $x\oplus y=x+y~({\rm mod}~s_m)$.

\begin{construction}\label{construction1}
\it
Let $F$ be a finite field with $|F|\geq s_mn$, and $\lambda_{i,j},i\in[n],j\in[0,s_m-1]$ be $s_mn$ distinct elements in $F$.
	The $(n,k,l=s\cdot s_m^n)$ MDS array code $\mathcal{C}_1$ is defined by the following $r\ell$ parity check equations
\begin{equation}\label{def-C1}
\sum_{i=1}^n\lambda_{i,a_i}^{t-1}c_{i,(a,b)}=0,~~a\in[0,s_m^n-1], b\in[0,s-1], t\in[r],
\end{equation}
where $r=n-k$.
\end{construction}

It is easy to see that for every $(a,b)\in[0,\ell-1]$, the vector $(c_{1,(a,b)},c_{2,(a,b)},\ldots,c_{n,(a,b)})$ forms an $[n,k]$ generalized Reed-Solomon (GRS) codeword. Thus $\mathcal{C}_1$ is an $(n,k,\ell)$ MDS array code.
Next we show the repair property of $\mathcal{C}_1$.

\begin{theorem}\label{thm1}
$\mathcal{C}_1$ satisfies the $(1,d_i)$-optimal repair property for all $d_i$, $i\in[m]$.
\end{theorem}

\begin{proof}
We first prove the $(1,d_m)$-optimal repair property of $\mathcal{C}_1$.
Suppose node $e$ fails, i.e., $\bm{c}_e$ is erased for some $e\in[n]$. Let $\mathcal{R}_m\subseteq[n]\setminus\{e\}$ with $|\mathcal{R}_m|=d_m$ be the set of helper nodes.
Fix some $a_j\in[0,s_m-1]$ for all $j\in[n]\setminus\{e\}$ and some $b\in[0,s-1]$. Set $a=(a_1,\ldots,a_{e-1},0,a_{e+1},\ldots,a_{n})$. Then $\forall v\in[0,s_m-1]$, the $(a(e,v), b)$-th parity check equation in (\ref{def-C1}) gives
\begin{equation}\label{repair-1}
\sum_{j=1}^n\lambda_{j,a_j}^{t-1}c_{j,(a(e,v),b)}=0,~~~t\in[r].
\end{equation}
For $t\in[r]$, summing up the $s_m$ equations in (\ref{repair-1}) on $v=0,1,\ldots,s_m-1$, one can obtain
\begin{equation}\label{repair-2}
\sum_{v=0}^{s_m-1}\lambda_{e,v}^{t-1}c_{e,(a(e,v),b)}+\sum_{j\neq e}\lambda_{j,a_j}^{t-1}\sum_{v=0}^{s_m-1}c_{j,(a(e,v),b)}=0.
\end{equation}
This implies the symbols $\{c_{e,(a(e,v),b)}: v\in[0,s_m-1]\}\cup\{\sum_{v=0}^{s_m-1}c_{j,(a(e,v),b)}: j\in[n]\setminus\{e\}\}$ constitute an $[s_m+n-1=d_m+r, d_m]$ GRS codeword, since the coefficients
$\lambda_{i,j}$'s are all distinct. Thus node $e$ can recover $\{c_{e,(a(e,v),b)}: v\in[0,s_m-1]\}$ by downloading the symbol $\sum_{v=0}^{s_m-1}c_{j,(a(e,v),b)}$ from each helper node $j\in\mathcal{R}_m$.
When $a_j$, $j\in[n]\setminus\{e\}$ run though all integers in $[0,s_m-1]$ and $b$ runs through $[0,s-1]$, node $e$ can be recovered by downloading in total $d_ms_m^{n-1}s=\frac{d_m\ell}{s_m}=\frac{d_m\ell}{d_m-k+1}$ symbols, achieving the cut-set bound (\ref{cut-set-bound}).

Next we prove the $(1,d_i)$-optimal repair property of $\mathcal{C}_1$ for $i\in[m-1]$. Suppose node $e$ fails and $\mathcal{R}_i\subseteq[n]\setminus\{e\}$ with $|\mathcal{R}_i|=d_i$ is the helper node set.
Denote $u_i=\frac{s}{s_i}$ for simplicity.
Fix some $a\in[0,s_m^n-1]$ and $\mu\in[0,u_i-1]$. For $t\in[r]$, summing up the $s_i$ parity check equations in (\ref{def-C1}) indexed by $(a(e,a_e\oplus v), b=\mu s_i+v)$ on $v=0,1,\ldots, s_i-1$, it has
$$
\sum_{v=0}^{s_i-1}\lambda_{e,a_e\oplus v}^{t-1}c_{e,(a(e,a_e\oplus v),\mu s_i+v)}+\sum_{j\neq e}\lambda_{j,a_j}^{t-1}\sum_{v=0}^{s_i-1}c_{j,(a(e,a_e\oplus v),\mu s_i+v)}=0.
$$
Similar to (\ref{repair-2}), it has that $\{c_{e,(a(e,a_e\oplus v),\mu s_i+v)}: v\in[0,s_i-1]\}\cup\{\sum_{v=0}^{s_i-1}c_{j,(a(e,a_e\oplus v),\mu s_i+v)}: j\in[n]\setminus\{e\}\}$ constitute an $[s_i+n-1=d_i+r, d_i]$ GRS codeword.
Then $\{c_{e,(a(e,a_e\oplus v),\mu s_i+v)}: v\in[0,s_i-1]\}$ can be computed from $d_i$ downloaded symbols in helper nodes. When $a$ runs through $[0,s_m^n-1]$ and $\mu$ runs through $[0,u_i-1]$, then node $e$ can be recovered by noticing the fact that $\{(a(e,a_e\oplus v),\mu s_i+v): a\in[0,s_m^n-1], \mu\in[0,u_i-1], v\in[0,s_i-1]\}=[0,\ell-1]$. The repair bandwidth equals $d_is_m^n\cdot\frac{s}{s_i}=\frac{d_i\ell}{d_i-k+1}$ symbols, achieving the cut-set bound (\ref{cut-set-bound}).

\end{proof}

\begin{remark}\label{remark0}
In \cite{Liu-2023-multi-d}, the authors give an MSR code with multiple repair degrees $d_i$, $i\in[m]$.
Denote $s={\rm lcm}(d_i-k+1)_{i\in[m]}$ and $s_1=\min(d_i-k+1)_{i\in[m]}$, the sub-packetization $\ell=s^{\lceil\frac{n}{s_1}\rceil}$ when $s_1\in\{2,3,4\}$ and $\ell=({\rm lcm}(4,s))^{\lceil\frac{n}{4}\rceil}$ when $s_i>4$ for all $i\in[m]$. The new code $\mathcal{C}_1$ in this section has larger sub-packetization than that in \cite{Liu-2023-multi-d}, however, it can serve as the base code for constructing centralized MSR codes with multiple node erasures and multiple repair degrees simultaneously.
\end{remark}

\section{Centralized MSR codes with multiple node erasures and multiple repair degrees simultaneously}\label{sec-code2}
In this section, we extend the code in Section \ref{sec-code1} to a centralized MSR code with multiple node erasures and multiple repair degrees simultaneously. Let $1\leq h_i\leq n-k$ and $k\leq d_i\leq n-h_i$ for $i\in[m]$.
Suppose $h_i\mid(d_i-k)$ and denote $s_i=\frac{d_i-k+h_i}{h_i}$ for $i\in[m]$. W.L.O.G., we can assume $s_1\leq s_2\leq\cdots\leq s_m$. Denote $s={\rm lcm}(s_1,\ldots,s_{m-1})$.
Then the code $\mathcal{C}_2$ is an $(n,k,\ell=s\cdot s_m^n)$ MDS array code with $(h_i, d_i)$-optimal repair property for all $i\in[m]$.

Using notations in Section \ref{sec-code1}, $\mathcal{C}_2$ can be defined similarly as in Construction \ref{construction1} using $r\ell$ parity check equations as in (\ref{def-C1}) by replacing $s_i=\frac{d_i-k+h_i}{h_i}$ for $i\in[m]$. Obviously, $\mathcal{C}_2$ is an MDS array code. We illustrate the repair property of $\mathcal{C}_2$ in the following Theorem.

\begin{theorem}\label{thm2}
$\mathcal{C}_2$ satisfies the $(h_i,d_i)$-optimal repair property for all $i\in[m]$.
\end{theorem}

\begin{proof}
We first prove the $(h_m,d_m)$-optimal repair property of $\mathcal{C}_2$.
W.L.O.G., suppose the $h_m$ failed nodes are $\mathcal{H}_m=\{1,2,\ldots, h_m\}$. Let $\mathcal{R}_m\subseteq[n]\setminus\mathcal{H}_m$ with $|\mathcal{R}_m|=d_m$ be the set of helper nodes.
Fix some $a_2,a_3,\ldots,a_n\in[0,s_m-1]$ and some $b\in[0,s-1]$.
Set $a=(0,a_2,a_3,\ldots,a_n)$. Let $\bm{1}_{h_m}$ denote the length-$h_m$ all-one vector $(1,1,\ldots,1)$.
$\forall v\in[0,s_m-1]$, $v\bm{1}_{h_m}=(v,v,\ldots,v)$.
Recall the notations in Section \ref{sec-code1}, $a(\mathcal{H}_m, a|_{\mathcal{H}_m}\oplus v\bm{1}_{h_m})=(v, a_2\oplus v, \ldots, a_{h_m}\oplus v, a_{h_m+1}, \ldots, a_n)$.
Now, consider the parity check equations indexed by $(a(\mathcal{H}_m, a|_{\mathcal{H}_m}\oplus v\bm{1}_{h_m}),b)$ for some $v\in[0,s_m-1]$, it has
\begin{equation}\label{repair-3}
\begin{aligned}
\sum_{j\in\mathcal{H}_m}\lambda_{j,a_j\oplus v}^{t-1}c_{j,(a(\mathcal{H}_m, a|_{\mathcal{H}_m}\oplus v\bm{1}_{h_m}),b)}+
\sum_{j\notin\mathcal{H}_m}\lambda_{j,a_j}^{t-1}c_{j,(a(\mathcal{H}_m, a|_{\mathcal{H}_m}\oplus v\bm{1}_{h_m}),b)}=0,~~~t\in[r].
\end{aligned}
\end{equation}
For $t\in[r]$, summing up the $s_m$ equations in (\ref{repair-3}) on $v=0,1,\ldots,s_m-1$, it has
\begin{equation}\label{repair-4}
\begin{aligned}
\sum_{j\in\mathcal{H}_m}\sum_{v=0}^{s_m-1}\lambda_{j,a_j\oplus v}^{t-1}c_{j,(a(\mathcal{H}_m, a|_{\mathcal{H}_m}\oplus v\bm{1}_{h_m}),b)}+
\sum_{j\notin\mathcal{H}_m}\lambda_{j,a_j}^{t-1}\sum_{v=0}^{s_m-1}c_{j,(a(\mathcal{H}_m, a|_{\mathcal{H}_m}\oplus v\bm{1}_{h_m}),b)}=0,~~~t\in[r].
\end{aligned}
\end{equation}
This implies the symbols
$$\begin{aligned}
\{c_{j,(a(\mathcal{H}_m, a|_{\mathcal{H}_m}\oplus v\bm{1}_{h_m}),b)}: j\in\mathcal{H}_m, v\in[0,s_m-1]\}
\cup\{\sum_{v=0}^{s_m-1}c_{j,(a(\mathcal{H}_m, a|_{\mathcal{H}_m}\oplus v\bm{1}_{h_m}),b)}: j\in[n]\setminus\mathcal{H}_m\}
\end{aligned}$$
constitute a $[h_ms_m+n-h_m=d_m+r,d_m]$ GRS codeword.
Thus, $\{c_{j,(a(\mathcal{H}_m, a|_{\mathcal{H}_m}\oplus v\bm{1}_{h_m}),b)}: j\in\mathcal{H}_m, v\in[0,s_m-1]\}$ can be recovered by downloading one symbol from each of the $d_m$ helper nodes $\mathcal{R}_m$. Let $a_2,\ldots,a_n$ run through all integers in $[0,s_m-1]$ and $b$ run through all integers in $[0,s-1]$, then the failed nodes in $\mathcal{H}_m$ can be repaired since the integer set
$\{(a(\mathcal{H}_m, a|_{\mathcal{H}_m}\oplus v\bm{1}_{h_m}),b): a=(0,a_2,\ldots,a_n), a_2,\ldots,a_n\in[0,s_m-1], b\in[0,s-1], v\in[0,s_m-1]\}$ is exactly the set of integers $[0,\ell-1]$.
The total repair bandwidth is $d_ms_m^{n-1}s=\frac{d_m\ell}{s_m}=\frac{d_mh_m\ell}{d_m-k+h_m}$ symbols, achieving the cut-set bound (\ref{cut-set-bound}).

Next, we prove the $(h_i,d_i)$-optimal repair property of $\mathcal{C}_2$ for all $i\in[m-1]$.
Given $i\in[m-1]$,
suppose the nodes in $\mathcal{H}_i$ with $|\mathcal{H}_i|=h_i$ are failed, and $\mathcal{R}_i\subseteq[n]\setminus\mathcal{H}_i$ with $|\mathcal{R}_i|=d_i$ is the helper node set.
Denote $u_i=\frac{s}{s_i}$ for simplicity. Fix some $a\in[0,s_m^n-1]$ and $\mu\in[0,u_i-1]$. Summing up the $s_i$ parity check equations indexed by $(a(\mathcal{H}_i, a|_{\mathcal{H}_i}\oplus v\bm{1}_{h_i}),b=\mu s_i+v)$
on $v=0,1,\ldots,s_i-1$, it has
$$
\begin{aligned}
\sum_{j\in\mathcal{H}_i}\sum_{v=0}^{s_i-1}\lambda_{j,a_j\oplus v}^{t-1}c_{j,(a(\mathcal{H}_i, a|_{\mathcal{H}_i}\oplus v\bm{1}_{h_i}),\mu s_i+v)}+
\sum_{j\notin\mathcal{H}_i}\lambda_{j,a_j}^{t-1}\sum_{v=0}^{s_i-1}c_{j,(a(\mathcal{H}_i, a|_{\mathcal{H}_i}\oplus v\bm{1}_{h_i}),\mu s_i+v)}=0,~~~t\in[r].
\end{aligned}
$$
Similar to (\ref{repair-4}), it indicates that the symbols
$$\begin{aligned}
\{c_{j,(a(\mathcal{H}_i, a|_{\mathcal{H}_i}\oplus v\bm{1}_{h_i}),\mu s_i+v)}: j\in\mathcal{H}_i, v\in[0,s_i-1]\}
\cup\{\sum_{v=0}^{s_i-1}c_{j,(a(\mathcal{H}_i, a|_{\mathcal{H}_i}\oplus v\bm{1}_{h_i}),\mu s_i+v)}: j\in[n]\setminus\mathcal{H}_i\}
\end{aligned}$$
forms a $[h_is_i+n-h_i=d_i+r,d_i]$ GRS codeword, and $\{c_{j,(a(\mathcal{H}_i, a|_{\mathcal{H}_i}\oplus v\bm{1}_{h_i}),\mu s_i+v)}: j\in\mathcal{H}_i, v\in[0,s_i-1]\}$ can be recovered from $d_i$ downloaded symbols.
Run through all $a\in[0,s_m^n-1]$ and all $\mu\in[0,u_i-1]$, the nodes in $\mathcal{H}_i$ can be repaired by noticing the fact that
$\{(a(\mathcal{H}_i, a|_{\mathcal{H}_i}\oplus v\bm{1}_{h_i}),\mu s_i+v): a\in[0,s_m^n-1], \mu\in[0,u_i-1], v\in[0,s_i-1]\}=[0,\ell-1]$. The repair bandwidth equals $d_is_m^nu_i=\frac{d_ih_i\ell}{d_i-k+h_i}$ symbols.

\end{proof}

According to Theorem \ref{thm2}, the following Corollary is straightforward.

\begin{corollary}\label{cor}
There is an $(n,k,\ell={\rm lcm}(1,2,\ldots,r)r^n)$ MDS array code satisfying $(h_i,d_i)$-optimal repair property for all possible $1\leq h_i\leq n-k$ and $k\leq d_i\leq n-h_i$ with $h_i\mid(d_i-k)$ simultaneously.
\end{corollary}

\section{Centralized MSR codes for general $h$ and $d$}\label{sec-code3}
The centralized MSR code $\mathcal{C}_2$ in section \ref{sec-code2} is designed for special values of $(h_i,d_i)$, i.e., $h_i\mid(d_i-k)$ for $i\in[m]$.
In this section, for a single pair $(h,d)$, we firstly present a code construction with $(h,d)$-optimal repair property for general $h$ and $d$. Then for special values of $h$ and $d$, we show the sub-packetization can be further reduced.

\subsection{Codes for general $h$ and $d$}
Let $2\leq h\leq n-k$ and $k<d\leq n-h$. Denote $\delta=\gcd(h,d-k)$, then $\delta\mid(d-k+h)$ and $\delta\mid(d-k+\delta)$.
The code $\mathcal{C}_3$ is an $(n,k,\ell=\frac{d-k+h}{\delta}(\frac{d-k+\delta}{\delta})^n)$ MDS array code over a finite field $F$.

For simplicity, denote $s_0=\frac{d-k+\delta}{\delta}$ and $s=\frac{d-k+h}{\delta}$, then $s=s_0+\frac{h}{\delta}-1$ and $\ell=s\cdot s_0^n$. Based on the notations as before, we use pairs $(a,b)$ with $a\in[0,s_0^n-1]$, $b\in[0,s-1]$ to represent integers in $[0,\ell-1]$ and use the notation $\oplus$ to denote addition modulo $s_0$.
Then the code $\mathcal{C}_3$ is defined similarly as in Construction \ref{construction1} using $r\ell$ parity check equations as in (\ref{def-C1}) by replacing $s_m=s_0$ and $s=s_0+\frac{h}{\delta}-1$.
Obviously, $\mathcal{C}_3$ is an MDS array code. In the following, we show the $(h,d)$-optimal repair property of $\mathcal{C}_3$.

\begin{theorem}\label{thm3}
$\mathcal{C}_3$ satisfies the $(h,d)$-optimal repair property.
\end{theorem}

We prove Theorem \ref{thm3} by giving the precise repair scheme.
Suppose $h$ nodes indexed by $\mathcal{H}=\{i_1,i_2,\ldots,i_h\}\subseteq[n]$ are failed and the helper node set is $\mathcal{R}\subseteq[n]\setminus\mathcal{H}$ with $|\mathcal{R}|=d$.
We partition the failed set $\mathcal{H}$ into $\frac{h}{\delta}$ disjoint subsets $P_1,\ldots, P_{h/\delta}$ with each $|P_i|=\delta$, $i\in[\frac{h}{\delta}]$.
 Moreover, for $i\in[\frac{h}{\delta}]$, define
$$
\Omega_i=\{0,1,\ldots,s_0-2, s_0-2+i\}.
$$
The repair process proceeds in two steps. In the first step, partial data of the failed nodes in each group $P_i$ for $i\in[\frac{h}{\delta}]$ is recovered in parallel. In the second step, the remaining data of each failed node is computed. We illustrate the first repair step in the following Lemma.

\begin{lemma}\label{lemma1} Denote by $\bm{1}_{\delta}$ the length-$\delta$ all-one vector.
For each $i\in[\frac{h}{\delta}]$,
By downloading the symbols
 \begin{equation}\label{downloaddata}
 \begin{aligned}
 \Big \{\sum_{v=0}^{s_0-2}c_{j,(a(P_i,a|_{P_i}\oplus v\bm{1}_{\delta}), v)}+
 c_{j,(a(P_i,a|_{P_i}\oplus (s_0-1)\bm{1}_{\delta}), s_0-2+i)}:
 0\leq a\leq s_0^n-1 \Big \}
 \end{aligned}
 \end{equation}
 from each helper node $j\in\mathcal{R}$, one can recover the data
\begin{equation}\label{recoverdata-1}
\cup_{j\in P_i}\big\{c_{j,(a,b)}: a\in[0,s_0^n-1], b\in\Omega_i\big\}
\end{equation}
 and
  \begin{equation}\label{recoverdata-2}
  \begin{aligned}
 \cup_{j\in \mathcal{H}\setminus P_i}\Big\{\sum_{v=0}^{s_0-2}c_{j,(a(P_i,a|_{P_i}\oplus v\bm{1}_{\delta}), v)}+
 c_{j,(a(P_i,a|_{P_i}\oplus (s_0-1)\bm{1}_{\delta}), s_0-2+i)}:
  a\in[0,s_0^n-1] \Big\}.
 \end{aligned}
 \end{equation}
\end{lemma}

\begin{proof}
Let $i\in[\frac{h}{\delta}]$, we show the repair process of the failed group $P_i$.
Fix some $a\in[0,s_0^n-1]$. For each $v\in[0,s_0-2]$, according to the parity check equations indexed by $(a(P_i,a|_{P_i}\oplus v\bm{1}_{\delta}), v)$, one has the following $r$ parity check equations, for $t\in[r]$,
\begin{equation}\label{pc-eq1}
\begin{aligned}
\sum_{j\in P_i}\lambda_{j,a_j\oplus v}^{t-1}c_{j,(a(P_i,a|_{P_i}\oplus v\bm{1}_{\delta}), v)}
+\sum_{j\in [n]\setminus P_i}\lambda_{j,a_j}^{t-1}c_{j,(a(P_i,a|_{P_i}\oplus v\bm{1}_{\delta}), v)}=0.
\end{aligned}
\end{equation}
Similarly, let $v=s_0-2+i$, according to the parity check equations indexed by $(a(P_i,a|_{P_i}\oplus (s_0-1)\bm{1}_{\delta}), s_0-2+i)$, one has $r$ parity check equations, for $t\in[r]$,
\begin{equation}\label{pc-eq2}
\begin{aligned}
\sum_{j\in P_i}\lambda_{j,a_j\oplus (s_0-1)}^{t-1}c_{j,(a(P_i,a|_{P_i}\oplus (s_0-1)\bm{1}_{\delta}), s_0-2+i)}
+\sum_{j\in [n]\setminus P_i}\lambda_{j,a_j}^{t-1}c_{j,(a(P_i,a|_{P_i}\oplus (s_0-1)\bm{1}_{\delta}), s_0-2+i)}=0.
\end{aligned}
\end{equation}
Summing up the $s_0$ equations (\ref{pc-eq1}) and (\ref{pc-eq2}) on all $v\in\Omega_i=[0,s_0-2]\cup\{s_0-2+i\}$, one can get $r$ parity check equations, shown in (\ref{pc-eq3}).
\begin{figure*}[!b]
   {\begin{equation}\begin{split}\label{pc-eq3}
&\sum_{j\in P_i}\Big(\sum_{v=0}^{s_0-2}\lambda_{j,a_j\oplus v}^{t-1}c_{j,(a(P_i,a|_{P_i}\oplus v\bm{1}_{\delta}), v)}
+\lambda_{j,a_j\oplus (s_0-1)}^{t-1}c_{j,(a(P_i,a|_{P_i}\oplus (s_0-1)\bm{1}_{\delta}), s_0-2+i)}\Big)\\
+&\sum_{j\in [n]\setminus P_i}\lambda_{j,a_j}^{t-1}\Big(\sum_{v=0}^{s_0-2}c_{j,(a(P_i,a|_{P_i}\oplus v\bm{1}_{\delta}), v)}
+c_{j,(a(P_i,a|_{P_i}\oplus (s_0-1)\bm{1}_{\delta}), s_0-2+i)}\Big)=0, ~~~~t\in[r].
\end{split}\end{equation}}
\end{figure*}

Note that in equation (\ref{pc-eq3}) the coefficients $\{\lambda_{j,a_j\oplus v}: j\in P_i, v\in[0,s_0-1]\}$ and $\{\lambda_{j,a_j}: j\in [n]\setminus P_i\}$ are all distinct, this implies that the $\delta s_0+n-\delta=d+r$ symbols
\begin{equation}\label{GRS-symbols}
\begin{aligned}&\big\{c_{j,(a(P_i,a|_{P_i}\oplus v\bm{1}_{\delta}), v)}: j\in P_i, v\in[0,s_0-2]\big\}\\
\cup&\big\{c_{j,(a(P_i,a|_{P_i}\oplus (s_0-1)\bm{1}_{\delta}), s_0-2+i)}: j\in P_i\big\}\\
\cup&\big\{ \sum_{v=0}^{s_0-2}c_{j,(a(P_i,a|_{P_i}\oplus v\bm{1}_{\delta}), v)}+
c_{j,(a(P_i,a|_{P_i}\oplus (s_0-1)\bm{1}_{\delta}), s_0-2+i)}: j\in [n]\setminus P_i \big\}
\end{aligned}
\end{equation}
constitute a $[d+r,d]$ GRS codeword.
Thus, the unknown data
\begin{equation*}
\begin{aligned}\cup_{j\in P_i}\big\{
 c_{j,(a(P_i,a|_{P_i}\oplus (s_0-1)\bm{1}_{\delta}), s_0-2+i)},
c_{j,(a(P_i,a|_{P_i}\oplus v\bm{1}_{\delta}), v)}: v\in[0,s_0-2]\big\}\\
\end{aligned}
\end{equation*}
and
\begin{equation*}
\begin{aligned}\cup_{j\in \mathcal{H}\setminus P_i}\big\{
 \sum_{v=0}^{s_0-2}c_{j,(a(P_i,a|_{P_i}\oplus v\bm{1}_{\delta}), v)}+
c_{j,(a(P_i,a|_{P_i}\oplus (s_0-1)\bm{1}_{\delta}), s_0-2+i)}  \big\}
\end{aligned}
\end{equation*}
can be recovered by downloading the symbol
$$\sum_{v=0}^{s_0-2}c_{j,(a(P_i,a|_{P_i}\oplus v\bm{1}_{\delta}), v)}+
c_{j,(a(P_i,a|_{P_i}\oplus (s_0-1)\bm{1}_{\delta}), s_0-2+i)}
$$
from each helper node $j\in \mathcal{R}$, where $|\mathcal{R}|=d$.

When $a$ runs through all the integers in $[0,s_0^n-1]$, then Lemma \ref{lemma1} can be obtained by noticing the fact that
\begin{equation*}
\begin{aligned}
&\big\{ (a(P_i,a|_{P_i}\oplus (s_0-1)\bm{1}_{\delta}), s_0-2+i) :  a\in[0,s_0^n-1] \big\} \\
\cup&\big\{ (a(P_i,a|_{P_i}\oplus v\bm{1}_{\delta}), v): v\in[0,s_0-2],a\in[0,s_0^n-1]   \big\} \\
=&\big\{ (a,b): a\in[0,s_0^n-1], b\in\Omega_i   \big\}
\end{aligned}
\end{equation*}

\end{proof}

Note that the repair bandwidth in the first repair step is $\frac{h}{\delta}ds_0^n=\frac{hd\ell}{d-k+h}$ symbols.
Then Theorem \ref{thm3} can be proved once the downloaded symbols in the first repair step are enough to repair all the $h$ failed nodes, indicating that the total repair bandwidth in the whole repair process is $\frac{hd\ell}{d-k+h}$ symbols, achieving the cut-set bound (\ref{cut-set-bound}).
In the following, we prove that the recovered data in (\ref{recoverdata-1}) and (\ref{recoverdata-2}) for all $i\in[\frac{h}{\delta}]$ as illustrated in Lemma \ref{lemma1} are able to recover all data in the failed nodes $\mathcal{H}$.

\begin{lemma}\label{lemma2}
All the data $\bm{c}_j$, $j\in\mathcal{H}$ can be computed from the recovered data (\ref{recoverdata-1}) and (\ref{recoverdata-2}) for $i\in[\frac{h}{\delta}]$ given in Lemma \ref{lemma1}.
\end{lemma}
\begin{proof}
By Lemma \ref{lemma1}, $\forall i\in[\frac{h}{\delta}]$, and $\forall j\in P_i$, the failed node $j$ has recovered some independent symbols:
$$
\big\{c_{j,(a,b)}: a\in[0,s_0^n-1], b\in\Omega_i\big\}
$$
and some symbol sums:
$$
\begin{aligned}
 \big\{c_{j,(a(P_{i'},a|_{P_{i'}}\oplus (s_0-1)\bm{1}_{\delta}), s_0-2+i')}+
 \underline{\sum_{v=0}^{s_0-2}c_{j,(a(P_{i'},a|_{P_{i'}}\oplus v\bm{1}_{\delta}), v)}}:
 a\in[0,s_0^n-1], i'\in[\frac{h}{\delta}]\!\setminus\!\{i\} \big\},
 \end{aligned}
$$
where the underlined part in the symbol sums can be computed from the independent symbols.
Thus node $j$ can further recover
$$
\big\{c_{j,(a(P_{i'},a|_{P_{i'}}\oplus (s_0\!-\!1)\bm{1}_{\delta}), s_0-2+i')}:a\in[0,s_0^n-1], i'\!\in\![\frac{h}{\delta}]\!\setminus\!\{i\} \big\},
$$
which is exactly the set of symbols
$$
\big\{c_{j,(a,b)}: a\in[0,s_0^n-1], b\in[0,s_0+\frac{h}{\delta}-2]\setminus\Omega_i\big\}.
$$
Thus the proof is completed.
\end{proof}

\begin{remark}\label{remark1}
When $\delta=1$, the code $\mathcal{C}_3$ reduces to the code construction in \cite{Ye2020} which is tailored for cooperative repair. However, when $\delta>1$, the code has much smaller sub-packetization than that in \cite{Ye2020} under the centralized repair model. The main technique for node repair in $\mathcal{C}_3$ follows the idea of ``local centralization and group collaboration". That is to say, the $h$ failed nodes are partitioned into $\frac{h}{\delta}$ local groups so that repair within each group can be seen as centralized and repair across groups can be seen as cooperative.
\end{remark}

\begin{remark}\label{remark2}
When $\delta=h$, i.e. $h\mid (d-k)$, the sub-packetization of $\mathcal{C}_3$ is $\ell=(\frac{d-k+h}{h})^{n+1}$. However, we claim that in this case, $\ell=(\frac{d-k+h}{h})^{n}$ is enough for an MDS array code to repair $h$ erasures.
Actually, recall the code construction $\mathcal{C}_2$ given in Section \ref{sec-code2}. Set $m=1$ and $s=1$. Then $\mathcal{C}_2$ has sub-packetization $\ell=(\frac{d-k+h}{h})^{n}$ and satisfies the $(h,d)$-optimal repair property.
This construction coincides with the code given in \cite{Li-arxiv}. Indeed, the works \cite{Ye2020} and \cite{Li-arxiv} can be seen as two extreme cases of our constructions with $\delta=1$ and $\delta=h$ respectively.

\end{remark}

\begin{remark}
The code $\mathcal{C}_3$ is designed for a single pair $(h,d)$. Following the idea that $\mathcal{C}_2$ is constructed, the code $\mathcal{C}_3$ can be extended to possess $(h_i,d_i)$-optimal repair property for all possible pairs $(h_i,d_i)$ simultaneously. We give the precise construction in Section \ref{sec-code4}.
\end{remark}

\subsection{Codes for special values of $h$ and $d$}
We continue to apply the repair idea to the Hadamard MSR code in  \cite{Liu2023} and modify the cooperative repair scheme to a centralized one with smaller sub-packetization.
Here we assume $d>k, (d-k)\mid h$  and $(\frac{h}{d-k}+1) \mid 2^n$. 
The Hadamard MSR code is an $(n,k,\ell=2^n)$ MDS array code over a finite field $F$ with $|F|>2n$.
As before, we write each integer $a\in[0,\ell-1]$ by its binary expansion representation $(a_1,\ldots,a_n)\in\{0,1\}^n$ where $a=\sum_{i=1}^{n}a_i2^{i-1}$.
The Hadamard MSR code is given by the following $r\ell$ parity check equations
\begin{equation}\label{hadamard}
\sum_{j\in[n]}\lambda_{j,a_j}^{t-1}c_{j,a}=0,  ~~~a\in[0,\ell-1],~ t\in[r],
\end{equation}
where $r=n-k$. Compared with the code $\mathcal{C}_3$, the sub-packetization is reduced by a factor of $\frac{h}{d-k}+1$.

Next we illustrate that the Hadamard MSR code enables $(h,d)$-optimal repair.
Let ${\rm Ham}(2,w)$ denote an $[N=2^w-1,K=2^w-1-w,3]$ Hamming code over $\mathbb{F}_2$.
Denote $V_0={\rm Ham}(2,w)$ and for $i\in[N]$, define 
\begin{equation}\label{V}
V_i=\{\bm{y}(i,y_i\oplus 1):~\bm{y}=(y_1,y_2,\ldots,y_N)\in V_0\},
\end{equation}
where ``$\oplus$'' means addition modulo 2. Since Hamming code is a perfect code and $V_0$ has minimum distance 3, one has the following Lemma.
\begin{lemma}\cite{Liu2023}\label{hamming}
$V_0,V_1,\ldots, V_N$ form a partition of the space $\mathbb{F}_2^N$, and $|V_i|=2^K=\frac{2^N}{N+1}$ for $i\in[0,N]$.
\end{lemma} 

Now suppose $h$ nodes indexed by $\mathcal{H}$ are failed, W.L.O.G, we can assume $\mathcal{H}=[h]$ and the helper node set is $\mathcal{R}\subseteq[n]\setminus[h]$ with $|\mathcal{R}|=d$. For simplicity, we denote $\delta={\rm gcd}(h,d-k)=d-k$ since $(d-k)\mid h$ here.
Partition the failed set $\mathcal{H}$ into $\frac{h}{\delta}$ disjoint subsets $P_1,\ldots, P_{h/\delta}$ with each $|P_i|=\delta$, $i\in[\frac{h}{\delta}]$.
Say, set $P_i=[(i-1)\delta+1, i\delta]$ for $i\in[\frac{h}{\delta}]$.
Since $(\frac{h}{\delta}+1) \mid 2^n$, then there exists some $w$ s.t. $\frac{h}{\delta}=2^w-1$.
Choose $V_0={\rm Ham}(2,w)$ to be the $[\frac{h}{\delta}, \frac{h}{\delta}-w, 3]$ Hamming code over $\mathbb{F}_2$, and $V_i, i\in[\frac{h}{\delta}]$ are defined as in (\ref{V}).
Denote $M=\{\delta, 2\delta, 3\delta, \ldots, h\}$ and $|M|=\frac{h}{\delta}$. Then for each $i\in\frac{h}{\delta}$, it has $P_i\cap M=\{i\delta\}$. The repair process is illustrated in the following Theorem.
\begin{theorem}\label{thm-hadamard}
The failed nodes in $\mathcal{H}$ can be repaired in a centralized way by downloading $\frac{h\ell}{d-k+h}$ symbols
$$\big\{
\sum_{v=0,1}c_{j,a(P_i, a|_{P_i}\oplus v\bm{1}_{\delta})}: ~a\in[0,\ell-1] ~{\rm with}~ a|_{M}\in V_0, ~i\in[\frac{h}{\delta}]
\big\}$$
from each helper node $j\in\mathcal{R}$. Thus the repair bandwidth is $\frac{dh\ell}{d-k+h}$ symbols, achieving the cut-set bound (\ref{cut-set-bound}).
\end{theorem}

\begin{proof}
$\forall i\in[\frac{h}{\delta}]$, we first recover partial data of the failed nodes in group $P_i$.
Fix some $a\in[0,\ell-1]$ with $a|_{M}\in V_0$. For $t\in[r]$, summing up the parity check equations in (\ref{hadamard}) indexed by $a$ and $a(P_i, a|_{P_i}\oplus\bm{1}_{\delta})$, one can get $r$ equations as follows.
$$
\sum_{j\in P_i}\sum_{v=0,1}\lambda_{j,a_j\oplus v}^{t-1}c_{j,a(P_i, a|_{P_i}\oplus v\bm{1}_{\delta})} + 
\sum_{j\in [n]\setminus P_i}\lambda_{j,a_j}^{t-1}\big(\sum_{v=0,1}c_{j,a(P_i, a|_{P_i}\oplus v\bm{1}_{\delta})}\big)=0.
$$
This implies the $2\delta+n-\delta=d+r$ symbols
$$
\{c_{j,a(P_i, a|_{P_i}\oplus v\bm{1}_{\delta})}: ~j\in P_i, v\in\{0,1\}\}\cup\{\sum_{v=0,1}c_{j,a(P_i, a|_{P_i}\oplus v\bm{1}_{\delta})}:~j\in[n]\setminus P_i\}
$$
constitute a $[d+r, d]$ GRS codeword. When $a$ runs through all integers in $[0,\ell-1]$ with $a|_{M}\in V_0$, one can get $\frac{2^n}{h/\delta+1}$ such GRS codewords.
Thus, the following unknown symbols
\begin{equation}\label{H-unknown1}
\cup_{j\in P_i}\{c_{j,a}: ~a\in[0,\ell-1] ~{\rm with}~ a|_{M}\in V_0\cup V_i\}
\end{equation}
and
\begin{equation}\label{H-unknown2}
\cup_{j\in \mathcal{H}\setminus P_i}\big\{\sum_{v=0,1}c_{j,a(P_i, a|_{P_i}\oplus v\bm{1}_{\delta})}: ~a\in[0,\ell-1] ~{\rm with}~ a|_{M}\in V_0\big\}
\end{equation}
can be recovered from the downloaded data since
\begin{equation}\label{H-seteq}
\{a(P_i, a|_{P_i}\oplus v\bm{1}_{\delta}): ~a\in[0,\ell-1] ~{\rm with}~ a|_{M}\in V_0, ~v\in\{0,1\}\}=\{a\in[0,\ell-1]:  ~a|_{M}\in V_0\cup V_i\}.
\end{equation}

Next we claim that the data in (\ref{H-unknown1}) and (\ref{H-unknown2}) for $i\in[\frac{h}{\delta}]$ are enough to recover all the failed data.
$\forall i\in[\frac{h}{\delta}]$ and $\forall j\in P_i$, one has recovered some independent symbols indicated by (\ref{H-unknown1}) and some symbol sums indicated by (\ref{H-unknown2}) with $P_i$ replacing by $P_{i'}$ for all $i'\in[\frac{h}{\delta}]\setminus\{i\}$.
Similar to Lemma \ref{lemma2} and according to (\ref{H-seteq}), one can compute the data 
$$
\{c_{j,a}: ~a|_{M}\in V_{i'} ~{\rm for ~all}~ i'\in[\frac{h}{\delta}]\setminus\{i\}\}
$$
Then the proof is concluded by noticing that 
$$
[0,\ell-1]=\cup_{i\in[0,\frac{h}{\delta}]}\{a\in[0,\ell-1]: ~a|_{M}\in V_i\}
$$
according to the definition of $V_i$'s and Lemma \ref{hamming}.

\end{proof}

\section{Centralized MSR codes with $(h,d)$-optimal repair for all possible $h$ and $d$ simultaneously}\label{sec-code4}
We present a centralized MSR code with $(h_i,d_i)$-optimal repair for general $1\leq h_i\leq n-k$ and $k\leq d_i\leq n-h_i$, $i\in[m]$, simultaneously.
Denote $\delta_i={\rm gcd}(h_i, d_i-k)$ and $s_i=\frac{d_i-k+\delta_i}{\delta_i}$ for $i\in[m]$. Then $\frac{d_i-k+h_i}{\delta_i}=s_i+\frac{h_i}{\delta_i}-1$, $i\in[m]$.
W.L.O.G, we can assume that $s_1\leq s_2\leq\cdots\leq s_m$.
Denote $s={\rm lcm}(\frac{d_1-k+h_1}{\delta_1},\ldots,\frac{d_m-k+h_m}{\delta_m})={\rm lcm}(s_1+\frac{h_1}{\delta_1}-1,\ldots,s_m+\frac{h_m}{\delta_m}-1)$.
 Then the code $\mathcal{C}_4$ is an $(n,k,\ell=s\cdot s_m^n)$ MDS array code over a finite filed $F$. Using notations as in Section \ref{sec-code1}, the code $\mathcal{C}_4$ is defined similarly as in Construction \ref{construction1}.

We mainly illustrate the repair property of $\mathcal{C}_4$ in the following.

\begin{theorem}
$\mathcal{C}_4$ has $(h_i,d_i)$-optimal repair property for all $i\in[m]$.
\end{theorem}
\begin{proof}
We give a sketch of the proof.
Given $i\in[m]$, we prove $\mathcal{C}_4$ has $(h_i,d_i)$-optimal repair property. For simplicity, denote $u_i=s/(s_i+\frac{h_i}{\delta_i}-1)$. We partition the set $[0,s-1]$ into $u_i$ disjoint subsets: for $\mu\in[u_i]$, define
$$\Gamma_\mu=\big\{(\mu-1)(s_i+\frac{h_i}{\delta_i}-1),(\mu-1)(s_i+\frac{h_i}{\delta_i}-1)+1,\ldots, (\mu-1)(s_i+\frac{h_i}{\delta_i}-1))+s_i+\frac{h_i}{\delta_i}-2\big\}\subseteq[0,s-1].$$
Then $|\Gamma_\mu|=s_i+\frac{h_i}{\delta_i}-1$ for $\mu\in[u_i]$.
For a codeword $\bm{c}=(\bm{c}_1,\ldots,\bm{c}_n)\in\mathcal{C}_4$ where each $\bm{c}_j=(c_{j,0},\ldots, c_{j,\ell-1})\in F^\ell$, we write $\bm{c}_j=(c_{j, (a,b)})_{a\in[0,s_m^n-1], b\in[0,s-1]}$ for $j\in[n]$.
For each $\mu\in[u_i]$, we apply the repair scheme in Section \ref{sec-code3} to $\bm{c}$ with each $\bm{c}_j$ restricted to the digits $(a,b)\in[0,s_m^n-1]\times \Gamma_\mu$.
Then one can verify that $\bm{c}$ enables optimal repair of any $h_i$ nodes using any $d_i$ helper nodes with repair bandwidth achieving the cut-set bound (\ref{cut-set-bound}).
\end{proof}

The following corollary can be directly obtained.
\begin{corollary}\label{cor2}
There is an $(n,k,\ell={\rm lcm}(1,2,\ldots,n-k)(n-k)^n)$ MDS array code satisfying $(h_i,d_i)$-optimal repair property for all possible $1\leq h_i\leq n-k$ and $k\leq d_i\leq n-h_i$ simultaneously.
\end{corollary}

\appendices


\begin{thebibliography}{100}

\bibitem{Dimakis2011}
A.~G. Dimakis, P.~G. Godfrey, Y.~Wu, and M.~O. Wainwright, K.~ Ramchandran,
\newblock ``Network Coding for Distributed Storage Systems,"
\newblock  {\it IEEE Trans. Inf. Theory}, vol. 56, no. 9, pp. 4539-4551, 2010.

\bibitem{arraycode}
M. Blaum, P. G. Farell, and H. van Tilborg, ``Array codes,'' {\it in Handbook of Coding Theory}, V. Pless and W. C. Huffman, Eds. Elsevier Science, 1998,
vol. II, ch. 22, pp. 1855-1909.

\bibitem{Cadambe2013}
V. R. Cadambe, S. A. Jafar, H. Maleki, K. Ramchandran, and C. Suh, ``Asymptotic interference alignment for optimal repair of MDS
codes in distributed storage,'' {\it IEEE Trans. Inf. Theory}, vol. 59, no. 5, pp. 2974-2987, 2013.

\bibitem{Rashmi}
 K. V. Rashmi, N. B. Shah, and P. V. Kumar, ``Optimal exact-regenerating codes for distributed storage at the MSR and MBR points via a product-matrix
construction,'' {\it IEEE Trans. Inf. Theory}, vol. 57, no. 8, pp. 5227-5239, 2011.

\bibitem{zigzag}
I. Tamo, Z. Wang, and J. Bruck, ``Zigzag codes: MDS array codes with optimal rebuilding,'' {\it IEEE Trans. Inf. Theory}, vol. 59, no. 3, pp. 1597-1616, 2013.

\bibitem{Wang2016}
Z. Wang, I. Tamo, and J. Bruck, ``Explicit minimum storage regenerating codes,'' {\it IEEE Trans. Inf. Theory}, vol. 62, no. 8, pp. 4466-4480, 2016.

\bibitem{Ye-highrate}
M. Ye and A. Barg, ``Explicit constructions of high-rate MDS array codes with optimal repair bandwidth,'' {\it IEEE Trans. Inf. Theory}, vol. 63, no. 4,
pp. 2001-2014, 2017.




\bibitem{Ye-access}
M. Ye and A. Barg, ``Explicit constructions of optimal-access MDS codes with nearly optimal sub-packetization,'' {\it IEEE Trans. Inf. Theory}, vol. 63,
no. 10, pp. 6307-6317, 2017.

\bibitem{Li-2017}
J. Li, X. Tang, and C. Tian, ``A generic transformation for optimal repair bandwidth and rebuilding access in MDS codes,'' {\it in Proc. IEEE Int. Symp.
Inf. Theory (ISIT)}, June 2017, pp. 1623-1627.

\bibitem{Liu}
Y. Liu, J. Li, and X. Tang, ``A generic transformation to generate MDS array codes with $\delta$-optimal access property,'' {\it IEEE Trans. Commun.}, vol. 70, no.
2, pp. 759-768, 2022.

\bibitem{Rawat}
A. S. Rawat, I. Tamo, V. Guruswami, and K. Efremenko, ``MDS code constructions with small sub-packetization and near-optimal repair bandwidth,''
{\it IEEE Trans. Inf. Theory}, vol. 64, no. 10, pp. 6506-6525, 2018.

\bibitem{Li2021}
J. Li, Y. Liu, and X. Tang, ``A systematic construction of MDS codes with small sub-packetization level and near-optimal repair bandwidth,'' {\it IEEE Trans.
Inf. Theory}, vol. 67, no. 4, pp. 2162-2180, 2021.

\bibitem{Guruswami-2020}
V. Guruswami, S. V. Lokam, and S. V. M. Jayaraman, ``$\epsilon$-MSR codes: Contacting fewer code blocks for exact repair,'' {\it IEEE Trans. Inf. Theory}, vol.
66, no. 11, pp. 6749-6761, 2020.



\bibitem{Liu-2023-multi-d}
Y. Liu, J. Li, and X. Tang, ``A generic transformation to enable optimal repair/access MDS array codes with multiple repair degrees,'' {\it IEEE Trans. Inf. Theory}, vol. 69, no. 7,
pp. 4407-4428, 2023.


\bibitem{Guruswami-rs}
V. Guruswami and M. Wootters, ``Repairing Reed-Solomon codes,'' {\it in Proceedings of Annual Symposium on Theory of Computing (STOC)}, 2016.

\bibitem{scalarMDS}
J. Mardia, B. Bartan and M. Wootters, ``Repairing multiple failures for scalar MDS codes'', {\it IEEE Trans. Inf. Theory}, vol. 65, no. 5,
pp. 2661-2672, 2019.

\bibitem{TotalRecall}
R. Bhagwan, K. Tati, Y. Cheng, S. Savage, and G. Voelker, ``Total recall: system support for automated availability management,'' {\it in Proc. of the 1st
Conf. on Networked Systems Design and Implementation}, Mar. 2004.

\bibitem{Wang-centralized}
Z. Wang, I. Tamo, J. Bruck, ``Optimal Rebuilding of Multiple Erasures in MDS Codes,'' {\it IEEE Trans. Inf. Theory}, vol. 63, no. 2, pp. 1084-1101, 2017.

\bibitem{Hu-cooperative}
Y. Hu, Y. Xu, X. Wang, C. Zhan, and P. Li, ``Cooperative recovery of distributed storage systems from multiple losses with network coding,'' {\it in IEEE
J. on Selected Areas in Commun}, vol. 28, no. 2, pp. 268-275, 2010.

\bibitem{Shum}
 Kenneth W. Shum, ``Cooperative regenerating codes for distributed storage systems,'' {\it in IEEE Int. Conf. Comm. (ICC)}, Kyoto, Jun. 2011.


\bibitem{erroneous}
 S. Pawar, S. El Rouayheb, and K. Ramchandran, ``Securing dynamic distributed storage systems against eavesdropping and adversarial attacks,'' {\it IEEE Trans. Inf. Theory}, vol. 57, no. 10, pp. 6734-6753, 2011.

\bibitem{Wang2019}
 M. Zorgui and Z. Wang, ``Centralized multi-node repair for regenerating codes,'' {\it IEEE Trans. Inf. Theory}, vol. 65, no. 7, pp. 4180-4206, 2019.

\bibitem{Li-arxiv}
S. Li, M. Gadouleau, J. Wang, D. Zheng, ``A New Centralized Multi-Node Repair Scheme of MSR codes with Error-Correcting Capability,'' {\it arXiv: 230915668}.

\bibitem{Ye-2018}
M. Ye, A. Barg, ``Cooperative repair: Constructions of optimal MDS codes for all admissible parameters'', {\it IEEE Trans. Inf. Theory}, vol. 65, no. 3, pp. 1639-1656, 2019.

\bibitem{Ye2020}
M. Ye, ``New constructions of cooperative MSR codes: reducing node size to exp($O(n)$),'' {\it IEEE Trans. Inf. Theory}, vol. 66, no. 12, pp. 7457-7464,
2020.

\bibitem{Liu2023}
Y. Liu, H. Cai, and X. Tang, ``A new cooperative repair scheme with $k+1$ helper nodes for $(n,k)$ Hadamard MSR codes with small sub-packetization,''
{\it IEEE Trans. Inf. Theory}, vol. 69, no. 5, pp. 2820-2829, 2023.


\end{thebibliography}
\end{document}